\documentclass[runningheads,a4paper]{llncs}

\usepackage{amssymb}
\usepackage{graphicx}

\usepackage{ulem}

\usepackage{lipsum}
\usepackage{multirow}

\usepackage{amsmath,amssymb}
\usepackage{pifont}
\newcommand{\xmark}{\ding{55}}

\usepackage{soul}
\usepackage{fancyhdr}
\usepackage{hyperref}
\usepackage[usenames,dvipsnames]{pstricks}
\usepackage{epsfig}
\usepackage{subfloat}
\usepackage{textcomp}
\usepackage{hyperref} 
\usepackage{url}
\usepackage{cite}
\usepackage{amsfonts}
\usepackage{url}
\urldef{\mailsa}\path|koksalmus@aydin.edu.tr,mehmet.kiraz@tubitak.gov.tr,|
\urldef{\mailsb}\path|mcenk@metu.edu.tr,isa.sertkaya@tubitak.gov.tr|
\newcommand{\keywords}[1]{\par\addvspace\baselineskip
\noindent\keywordname\enspace\ignorespaces#1}

\begin{document}

\mainmatter  

\title{Estonian Voting Verification Mechanism Revisited}

\titlerunning{Estonian Voting Verification Mechanism Revisited}

%
%

\author{K\"{o}ksal Mu\c{s}$^{1}$ \and Mehmet Sab{\i}r Kiraz$^{2}$ \and Murat Cenk$^{3}$ \and \.{I}sa Sertkaya$^{2}$}
\authorrunning{Estonian Voting Verification Mechanism Revisited}

\institute{$^{1}$\.{I}stanbul Ayd{\i}n University, $^{2}$T\"UB\.{I}TAK B\.{I}LGEM \\$^{3}$Middle East Technical University Turkey
\mailsa\\
\mailsb\\ 
}


%
%

\toctitle{Lecture Notes in Computer Science}
\tocauthor{Authors' Instructions}
\maketitle

\begin{abstract}
	
After the Estonian Parliamentary Elections held in 2011, an additional verification mechanism was integrated into the i-voting system in order to resist corrupted voting devices, including the so called \textsf{Student's Attack} where a student practically showed that the voting system is indeed not verifiable by developing several versions of malware capable of blocking or even changing the vote. This mechanism gives voters the opportunity to verify whether the vote they cast is stored in the central system correctly. However, the verification phase ends by displaying the cast vote in plain form on the verification device. In other words, the device on which the verification is done learns the voter's choice. In this work, our aim is to investigate this verification phase in detail and to point out that leaking the voter's choice to the verification application may harm the voter privacy. Additionally, when applied in a wide range, this would even compromise the fairness and the overall secrecy of the elections. In this respect, we propose an alternative verification mechanism for the Estonian i-voting system to overcome this vulnerability. Not only is the proposed mechanism secure and resistant against corrupted verification devices, so does it successfully verify whether the vote is correctly stored in the system. We also highlight that our proposed mechanism brings only symmetric encryptions and hash functions on the verification device, thereby mitigating these weaknesses in an efficient way with a negligible cost. More concretely, it brings only $m$ additional symmetric key decryptions to the verification device, where $m$ denoting the number of candidates. Finally, we prove the security of the proposed verification mechanism and compare the cost complexity of the proposed method with that of the current mechanism.
	
\end{abstract}

\keywords{Internet Voting, Privacy, Verifiability, Trust}

\section{Introduction} \label{intro}

Technology is frequently used in daily routines for governmental or banking services via the internet using computers or smart devices. Among these services, internet voting (i-voting) has the potential of increasing election participation, allowing voters, especially handicapped citizens or those citizens living abroad, to cast a vote without going to polling stations on a specific day or time. However, related security issues have not been taken into extensive consideration when the users install applications onto their devices. In particular, by not paying attention to the permissions given to the applications, users turn their smart devices into potential targets for malicious malware that may be used to obtain critical information about users \cite{SmartDevice2014}. 

Estonian i-voting protocol with its verification mechanism present an interesting case because it avoids the additional pre- and post-channels as seen in the Norwegian protocol, in which the verification is performed via smart devices \cite{SEC14, Meter2015a}. Since 2005 Estonian i-voting system is still being used and the number of i-voters increases in every election. In 2005 local election trial, while only 1.9\% of all votes were cast using the i-voting system, more specifically, 5.5\%, 14.7\%, 15.8\%, 24.3\%, 21.2\%, 31.3\% and 30.5\% of votes were cast using the i-voting system in the upheld elections successively \cite{statistics}. These statistics show that the increasing number of citizens prefer to use i-voting system. Accordingly, security concerns related the i-voting system should be considered more seriously.

The i-voting system aims to be at least as secure as traditional paper ballots, meaning that i-voting should meet both cast-as-intended and recorded-as-cast requirements \cite{E2E,EstSecAn-2010}. As mentioned in \cite{VRF14, APP12-10Pg}, client-side weaknesses were experienced in both the cast-as-intended and recorded-as-cast mechanisms during Estonia's 2011 parliamentary elections, so called {\sf Student's Attack}. Therefore, after the 2011 election, a verification mechanism was added to the system that gives voters the opportunity to verify the vote stored in the system via a smart device with a camera and internet connection. The verification mechanism pilot was first tested in the 2013 local elections and then used in the European Parliament Elections in 2014, and the Parliamentary Elections in 2015. Although using an application on a smart device for voting verification solves the aimed security weaknesses, it may bring with it additional problems related not only to the voter privacy, but also to the secrecy of election results.

\subsection{Contributions} 

In this work, we point out an important privacy issue in the verification mechanism of the Estonian i-voting system. The motivation of our attack comes from the fact that all voter details including the real vote are displayed by the verification device. We stress that if the smart device running the verification application is corrupted, then vote privacy can be easily compromised by sniffing the process on the verification device of which the voter is most probably the owner. In the post-Snowden world, although there is a growing awareness and concern about security, privacy and integrity of data on our mobile devices, there is still considerable number of mobile users who install mobile applications without paying attention to their potential security or privacy issues. Therefore, assuming the corruption of a smart device is relevant due to the huge number of increase in malwares during the last years \cite{FFCHW11,JS12,symantic,ZXHC14,Bach15,MobilePrivacy,MobilePrivacyLeakage,MobileMonitoring,MobileAttacks}. Hence, it is possible for an adversary to acquire an IMEI number and other private information, such as location, contacts, phone number, emails, and photos from smart devices including voting details, thereby compromises the voters' privacy. 

The goal of this paper is to mitigate the described privacy leakage of the Estonian i-voting verification mechanism. In this respect, we propose a new, privacy-preserving, and an efficient verification mechanism even if a corrupted verification device is used. Our proposal is quite practical since only a few additional symmetric encryptions on the verification device is performed. Secondly, the secrecy of the election results may also be violated within a wide range attacks. Specifically, an attacker may obtain information about the partial results of the election before it has ended \cite{VRF14}. In this work, however, our updated verification mechanism ensures the same security level without leaking any information about the vote. 

\subsection{Organization} The rest of the paper is organized as follows. Related works are discussed in the following section. The necessary preliminaries and underlying cryptographic mechanisms are explained in Section \ref{preliminaries}. The current Estonian i-voting system along with its components, security and threat models are explained in Section \ref{estonian}. A new potential privacy issue, the proposed voting protocol, and its security model are given in Section \ref{proposedsystem}. The security analysis of the proposed system is presented in Section \ref{newsecurity}. Section \ref{complexity} compares the complexities of Estonian system with our proposed system. Finally, Section \ref{conclusion} concludes the paper.

\section{Previous Work}\label{previous}

In this section, we will give a brief overview of related work, focusing on only internet voting schemes. After the 2011 elections in Estonia, Heiberg et al. published a paper discussing new attacks and weaknesses resulting from client and server side weaknesses namely {\sf Student's Attack} \cite{APP12-10Pg, VRF14, Hal15, Est15}. The designers of the Estonian i-voting system improved it by adding a verification mechanism. Like in the Norwegian i-voting scheme, using SMS services as a post channel was a possible solution; however, not all citizens may register their mobile numbers. Furthermore, the post channel mechanism was not only rather expensive, it also had various problems, as already seen in the Norwegian election system \cite{NRW2012, Meter2015a}. After a period of research and analysis, it was agreed that an individual verification mechanism using smart devices without requiring any personal information would be the most suitable verification channel for the Estonian i-voting system \cite{VRF14}. It also well-known that the Helios system is end-2-end verifiable which is not sufficient for secure elections since it does not prevent attacks from both corrupted client and servers \cite{Gs13, Halderman2015, BRR15, GKVWZ16}.

The designers of the Estonian i-voting system claim that it was as reliable and secure as the conventional election \cite{SEC14-20, Meter2015a}. Contrary to their security claims, in \cite{SEC14}, Springall et al. reported that the system is plagued by serious procedural and architectural weaknesses enabling client-side attacks that skew the results of the election undetectably bypassing the ID card system and smart device verification mechanism. Additionally, it is claimed that there are several inadequate procedural controls, lax operational security, insufficient transparency and several vulnerabilities in the published code. Moreover, in the same work, Springall et al. implemented a mock election in which they experienced both client and server side attacks. In responses the authors presented their recommendations on how to eliminate inadequate procedural controls and lax operational security weaknesses. In \cite{LOG15}, Heiberg et al. researched ways to eliminate transparency weaknesses using an auditing mechanism.  

After the NSA whistle-blowing revelations from Edward Snowden it is not easy to guess that the future brings more security and privacy risks for mobile devices \cite{FFCHW11,JS12,ZXHC14}. For example, users can be fooled into installing malicious applications on their devices or to grant unauthorized remote access \cite{FFCHW11,JS12,symantic,ZXHC14,Bach15,MobilePrivacy,MobilePrivacyLeakage,MobileMonitoring,MobileAttacks}. Hence, an adversary can easily identify the owner of the smart device via private information, such as one's IMEI number, location, contacts, phone number, emails, and photos. More specifically, IMEI numbers might also be required to be recorded into a central system beforehand which are used to identify and authenticate the mobile device whenever there is a connection request to a carrier. Those IMEI numbers, which are not recorded into the system, can be banned from communicating (e.g., \cite{turkey2016}). For these reasons, one should never be able to obtain any information about the intention of a voter from the voting details on the verification device .

\section{Preliminaries} \label{preliminaries}
In the next section, we will present the general setup and symbols needed to presenting our protocol.\\
\\
\textbf{Underlying Cryptosystem.} Denote a symmetric key encryption process as $\mathcal{E}_{\textsf{sym}}=\textsf{SymEnc}_{k}(M)$ and decryption as $M = \textsf{SymDec}_{k}(\mathcal{E}_{\textsf{sym}})$, where $k$ is a secret key and $M$ is a plaintext to be encrypted. We also denote the hashing of a message $M$ as $\textsf{H}(M)$, where $M$ is a message and $H : \{0,1\}^* \rightarrow \{0,1\}^t$. AES-256 and SHA3-256 (where $t=256$) can be utilized for symmetric encryption and hash function, respectively \cite{nist:2001a,nist:2015a}.

The underlying public key encryption scheme is semantically secure \footnote{Note that semantically secure cryptographic algorithms are basically randomized encryptions meaning that encryption of the same votes is uniformly indistinguishable from each other.} (e.g., Paillier \cite{paillier}, and ElGamal \cite{elgamal}). An election specific public and private key pair ($pk_{S}, sk_{S}$) is generated by the servers of the National Electoral Committee ({\sf NEC}) in a $k$  of $n$ threshold manner. In other words, it is generated by the cooperation of $n$ independent parties and it is not possible to regenerate the key pair if at least $k$ of them do not cooperate. Furthermore, $pk_{S}$ is mounted to voter applications \textsf{VoterApp}. Similarly, ($pk_{\mathcal{V}}, sk_{\mathcal{V}}$) denotes a public and private key pair of a user $\mathcal{V}$. $E_{\textsf{asym}}=\textsf{AsymEnc}_{pk_\mathcal{V}}(M)$ denotes an encryption of a message $M$ using the public key of the voter $\mathcal{V}$. Similarly, $\textsf{Sign}_{sk_\mathcal{V}}(M)$ denotes the signature of a user $\mathcal{V}$ on a message $M$ using the private key of $\mathcal{V}$.


\section{The Estonian Internet Voting Protocol and its Security Analysis}\label{estonian}

\subsection{Components of the System}

The Estonian i-voting mechanism is composed of three main parts: (1) client-side applications (a voter application \textsf{VoterApp} and a verification application \textsf{VerifApp}), (2) a {\sf Central System}, and (3) Auditing and Counting processes. 

A \textsf{VoterApp} is performed by the citizens eligible to vote via their ID cards. \textsf{VoterApp} is already developed and published by {\sf NEC}, and should be installed beforehand. {\sf VerifApp} should be installed on a smart device. Note that {\sf VerifApp} can be developed by any parties, including \textsf{NEC}, political parties, or an open source community. The {\sf Central System} has three main servers for forwarding, storing and counting phases under \textsf{NEC} responsibility. The Vote Forwarding Server ({\sf VFS}) is the server that the client-side applications authenticate, send signed and encrypted votes, and obtain the required data for both the voting and the verification stages. The Vote Storage Server ({\sf VSS}) stores the signed encrypted votes during the voting period. At the end of the election, it removes double votes, cancels ineligible voters, and prepares the votes to be tallied by anonymizing the encrypted votes via a mixnet mechanism (e.g., \cite{Mixnet2002}). 

Separated from the rest of the system by an air gap, the Vote Counting Server tallies the anonymized votes and computes the election results. During all these processes, the auditing mechanisms save logs of the {\sf Central System} events in order to resolve independent auditors' disputes and complaints.

\subsection{The Estonian I-Voting Protocol}
For simplicity of describing the voting protocol, we divide it into two phases: the voting phase and the verification phase. The voting phase begins by {\sf VoterApp} authenticating {\sf VFS} via a {\sf TLS connection}. A voter $\mathcal{V}$ receives the related candidate list $CL =\{c_1,\cdots,c_m\}$ where $c_i$'s are candidates' unique identity values and $m$ denotes the number of candidates. Next, the voter $\mathcal{V}$ chooses a candidate $c \in CL$ to cast the vote. {\sf VoterApp} generates a signed and encrypted vote $\textsf{SignEncVote}=\textsf{Sign}_{sk_{\mathcal{V}}}(E_{\textsf{asym}})$ where $ E_{\textsf{asym}} = \textsf{AsymEnc}_{pk_{S}}(c,r)$ and $r \in_R \{0,1\}^\kappa$ is a random number, $\kappa \in \mathbb{N}$. Next, {\sf VoterApp} sends \textsf{SignEncVote} to {\sf VFS}, and then receives a vote reference {\sf voteref} which is a receipt to be used in the verification phase.

During the verification phase, {\sf VerifApp} receives $r$ and {\sf voteref} from {\sf VoterApp}, request the related data from {\sf VFS} by the {\sf voteref}, and computes the vote. Finally, {\sf VerifApp} shows the recorded vote on the screen. If the voter confirms the correctness of the cast vote then the voting procedure ends successfully, otherwise, the voter $\mathcal{V}$ puts an alarm.

\subsubsection{Voting Stage:} 
\begin{enumerate}
	\item A voter $\mathcal{V}$: Authenticates to {\sf VFS} through {\sf VoterApp} using a national ID Card.
	\item {\sf VFS}: Sends $CL =\{c_1,\ldots,c_m\}$ to {\sf VoterApp} where $m$ is the number of candidates.
	\item $\mathcal{V}$: Chooses $c$ from $CL$ 
	\item {\sf VoterApp}:
	\begin{enumerate}
		\item Generates a random number $r$. 
		\item Encrypts $c$ and $r$ by $pk_{S}$,  $E_{\textsf{asym}}=\textsf{AsymEnc}_{pk_{S}}(c,r)$.
		\item Signs $E_{\textsf{asym}}$ by $sk_{\mathcal{V}}$, i.e. $\textsf{SignEncVote}=\textsf{Sign}_{sk_{\mathcal{V}}}(E_{\textsf{asym}})$ .
		\item Sends \textsf{SignEncVote} to {\sf VFS}.
	\end{enumerate}
	\item {\sf VFS}: 
	\begin{enumerate}
		\item Stores $\textsf{SignEncVote}$.
		\item Generates {\sf voteref}.
		\item Sends {\sf voteref} to {\sf VoterApp}.
	\end{enumerate}
\end{enumerate}

\subsubsection{Verification Stage:}
Note that the verification stage is optional and used only to ensure whether the vote has been correctly stored in {\sf VFS}. It is also important to note that for security purposes, {\sf VerifApp} and {\sf VoterApp} should not be installed on the same device. Additionally, it should be noted that {\sf VerifApp} scans the {\sf QR code} by camera instead of obtaining it via an internet connection. 

\begin{figure*}[htpb]
	\centering
	\caption{Verification stage of the Estonian i-voting protocol}
	\includegraphics[width=0.9\linewidth]{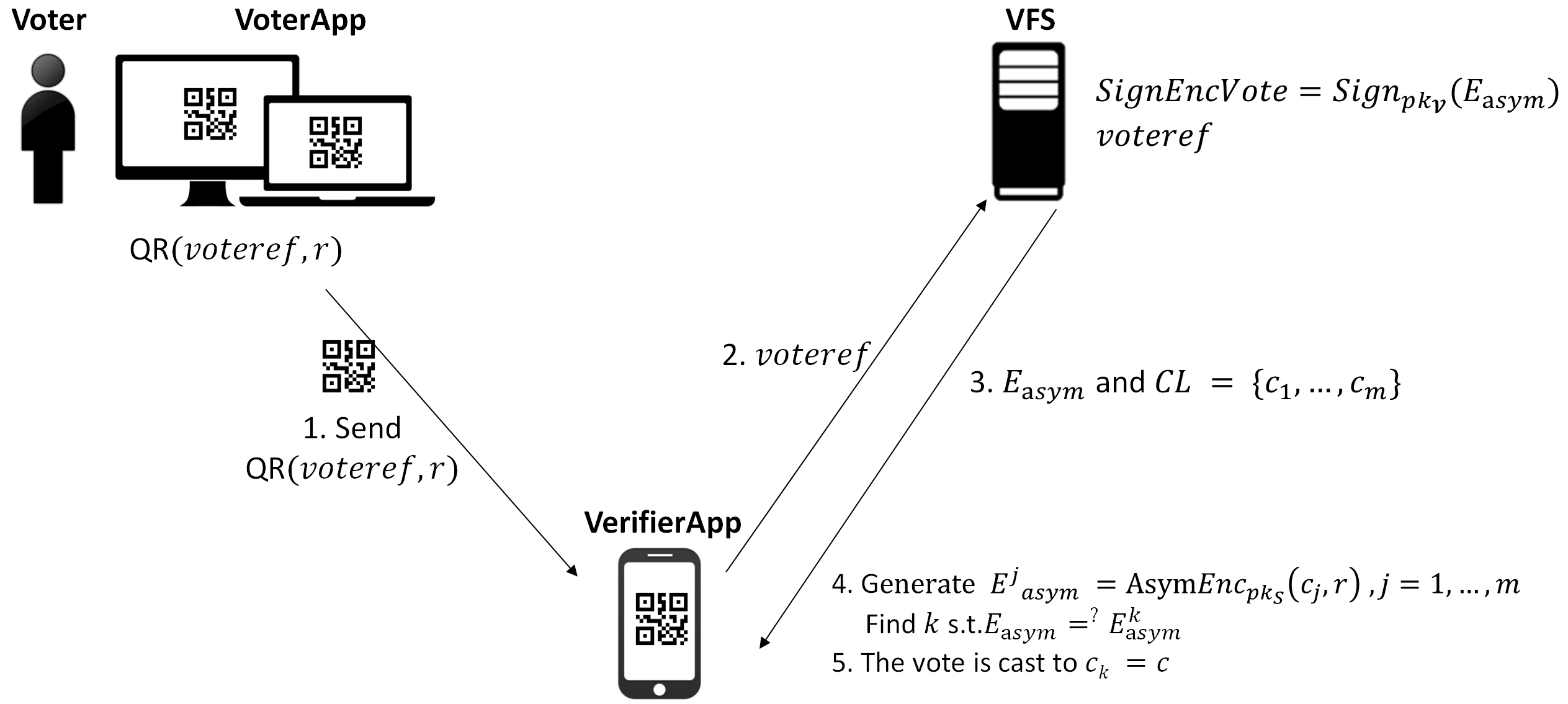}
	\label{fig:VerificationPhase_Orj}
\end{figure*}

\begin{enumerate}
	\item \begin{enumerate}
		\item {\sf VoterApp}: Generates a {\sf QR code} including $r$ and {\sf voteref} and show on the screen. 
		\item {\sf VerifApp}: Scans {\sf QR code} by camera. 
	\end{enumerate} 
	\item {\sf VerifApp}: Sends {\sf voteref} to {\sf VFS}.
	\item {\sf VFS}: Sends $E_{\textsf{asym}}$ and $CL$ to {\sf VerifApp}.
	\item {\sf VerifApp}: 
	\begin{enumerate}
		\item Computes $E_{\textsf{asym}}^j = \textsf{AsymEnc}_{pk_{S}}(c_j, r)$ for all $j=1,\cdots,m$. 
		\item Finds $\ell$ such that $E_{\textsf{asym}} \stackrel{?}{=} E_{\textsf{asym}}^{\ell}$ for some $\ell \in \{1,\cdots,m\}$. 
		\item Shows $c_{\ell}$ on the screen.
	\end{enumerate}
	\item $\mathcal{V}$: Checks $c_{\ell} \stackrel{?}{=} c$.
	\begin{enumerate}
		\item If $c_{\ell} = c$, the vote is received and stored {\sf VFS} without any modification. 
		\item Else, $\mathcal{V}$ puts an alarm (which basically shows that malware is present).
	\end{enumerate}
\end{enumerate}

\subsection{Security Model of the Estonian I-Voting Protocol} \label{security}

A security analysis of the current i-voting system of Estonia is discussed in detail in \cite{VRF14,SEC14} which are briefly mentioned in attack scenarios. Estonian i-voting security model assumes that either {\sf VoterApp} or the device that runs {\sf VoterApp} is malicious. We note that the assumptions {\sf VerifApp} and {\sf VFS} collude maliciously or {\sf VerifApp} and {\sf VoterApp} collude maliciously are not realistic since the duty of {\sf VerifApp} is to independently check the correctness of {\sf VFS} and {\sf VoterApp}. Furthermore, as noted in \cite{EstSecAn-2010}, limited number of corrupted voters' devices are accepted as a reasonable risk. 

\subsubsection{Attack Scenarios.} The main attack scenarios are about ballot integrity, the reliability of the voting system, and coercion resistance.  

\begin{itemize}
	\item{\textbf{Manipulation Attacks.}} Manipulation attacks consist of modifications to a vote without the knowledge of the voter $\mathcal{V}$. These attacks are aimed to change the vote to either a predetermined or a random candidate. There are basically two variants of manipulation attacks:
	
	\begin{itemize}
	\item{\sf Student's Attack.} In Estonia's 2011 parliamentary elections, the {\sf Student's Attack} exposed that neither ballot integrity and secrecy in the election was guaranteed \cite{APP12-10Pg}. The attack is based on installing a malware to the device that runs on {\sf VoterApp}. This malware is designed so that it undermines {\sf VoterApp} and while the vote is being cast, it silently diverts or cancels the intended vote. At first, this attack was not considered as a thread (see \cite{APP12-10Pg,SEC14}). But, after, it became a real and efficient attack \cite{APP12-10Pg}. So that, a verification mechanism is integrated into the system \cite{VRF14}.

	\item{\sf Ghost Click Attack.} A malicious software that runs on the same device as {\sf VoterApp} can obtain the {\sf PIN code} of the {\sf ID card} during the voting process. When the {\sf ID card} is reinserted, the malware may re-vote silently without being detected \cite{SEC14}. Although this is an interesting attack, it is not included in the scope of this work.  
	\end{itemize}
		
	\item{\textbf{Reliability Issues.}} 
	By decreasing voters' confidence in the i-voting system, an election can be held questionable without there being any violation to vote secrecy. Therefore, not only an effective verification mechanism must be set up to prevent such attacks,  but also, citizens should be well informed about the process and security concerns about the system \cite{EstSecAn-2010}. In fact, fairness and the secrecy of the elections should be satisfied, that is in any paper based or electronic voting scheme, the election's partial or total results can not be revealed before the tallying process. 
		
	\item{\textbf{Coercion Attacks.}} An i-voting system must prevent a voter from being able to prove to a coercer how he voted. Therefore, the system should provide \textsf{receipt-freeness} or \textsf{coercion resistance}. Allowing vote updates (i.e., re-voting) is the countermeasure to prevent such attacks.

\end{itemize}

\subsubsection{Threat Model.}

\begin{itemize}
	\item \textbf{Malicious {\sf VFS}.} In this case, there is a full privacy leakage where {\sf VFS} may leak all stored information as a result of potential {\sf Insider Attacks} \cite{SEC14}. As a countermeasure, auditing mechanisms with independent auditors can verify the computations of {\sf VFS} \cite{LOG15}. 
	
	\item \textbf{Malicious {\sf VoterApp} and its adversarial environment.} During the preparation of signed encrypted vote, a malicious {\sf VoterApp} can change the voter's intention $c \in CL$ by $c^* \in CL$ where $c \neq c^*$ and then sends encrypted form of $c^*$ instead of $c$. During the verification, {\sf VerifApp} will reveal that the ballot does not reflect the voter's own will as intended and puts an alarm. The voter $\mathcal{V}$ re-votes again from another device. Note that integrity of the stored data in {\sf VFS} is ensured after successful verification. Therefore, manipulating the vote using a malicious {\sf VoterApp} is not possible.
	
	\item \textbf{Malicious {\sf VerifApp} and its adversarial environment}: It is possible to sniff data from a smart device using any malicious application installed on the device \cite{FFCHW11,JS12,symantic,ZXHC14,Bach15,MobilePrivacy,MobilePrivacyLeakage,MobileMonitoring,MobileAttacks}. Leaking verified votes from the smart devices gives an opportunity to get some idea about the election results before the election ended. This attack violates one of the important limitation "secrecy of election results" of the Estonian i-voting system. Another attack based on the leakage is that using a device's IMEI number allows one to match the vote with the voter $\mathcal{V}$. In other words, the link between the vote and the voter $\mathcal{V}$ is revealed. The attack disrupts the privacy of $\mathcal{V}$ which is one of the main limitation of the Estonian i-voting system. Note that, the candidate choice of a voter represents the voter's political view. Therefore, when the voter's intention revealed, there is no way to recover. In other words, it is not possible to choose another candidate and vote again. More details about the attacks are given in \ref{potentialweakness}.
\end{itemize}

In the next section, we will show that there is a practical and realistic attack that may become a real issue. 

\section{A New Potential Privacy Issue in the Estonian Verification Protocol and Our Improved Verification System}\label{proposedsystem}

\subsection{A Vote Privacy Issue in the Estonian Verification Mechanism}\label{potentialweakness}

Recall that {\sf VerifApp} is integrated into the i-voting system to be resistant against {\sf Student's Attack}. {\sf VerifApp} can be freely developed by {\sf NEC}, by an open source community or even by one who is able to write her own verification software while {\sf VoterApp} is developed only by NEC. In contrast, most citizens cannot develop or control the trustability of {\sf VerifApp} even it is downloaded from a known source. Therefore, it is not realistic to assume that {\sf VerifApp} is honest. The Estonian i-voting system is designed in this manner and {\sf VerifApp} explicitly outputs the voters' intention in plain form.

Even if {\sf VerifApp} itself were to work properly, any malicious application running on the same device may sneak into the processes in an attempt to monitor inputs and outputs. Therefore, it would be sufficient to have any privileged application running on the same device \cite{FFCHW11,JS12,symantic,ZXHC14,Bach15,MobilePrivacy,MobilePrivacyLeakage,MobileMonitoring,MobileAttacks}. As soon as an adversary manages to install such a privileged application, the only thing that needs to be done is to grab a screen shot of the device while {\sf VerifApp} displays the output to the voter $\mathcal{V}$. After that, the cast vote will also be known by the adversary. Furthermore, since IMEI numbers and other private information (e.g, contacts, phone number, location, emails, photos) can be obtained by a malicious application loaded on the verification device, an adversary can obtain $\mathcal{V}$'s identity. Hence, privacy of $\mathcal{V}$ can be easily compromised. We highlight that this privacy issue may lead to coercion or vote buying (e.g., an adversary can force voters to install a malicious application on their smart devices for later check the their actual votes). 

When considered on the individual level, this attack causes privacy leakage. On the other hand, when applied over a wide range by sniffing the plain votes from the {\sf VerifApp}s without trying to find the related identity information, this attack may also promise the reliability and, more importantly, the secrecy of the elections. Because, as explained in \cite{EstOverview-2010}, in any election the results would not be exposed before the counting process has been partially or totally completed. Therefore, the possibility of a corrupted verification device must be included in the design criteria of the Estonian i-voting system, the current Estonian i-voting system requires additional countermeasures to guard against this privacy leakage. 

In the next section, we offer a new protocol for the verification mechanism which eliminates the trust to the verification device. 

\subsection{Our Enhanced Security Model}

Our security model extends the Estonian scheme by eliminating the need to trust the verification device on which {\sf VerifApp} runs. Namely, the extended scheme gives specific information to {\sf VerifApp} on which the voted candidate should be recognized only by the voter $\mathcal{V}$ with her own eyes manually. Even {\sf VerifApp} can not learn the voted candidate using the data which is gathered from the network or its outputs by {\sf VerifApp}. More formally, the probability of correctly guessing $\mathcal{V}$'s intention using the information given by {\sf VerifApp}'s inputs and outputs should be the same as the conditional probability of guessing $\mathcal{V}$'s intention randomly. We stress that using a mobile device to verify the vote makes the Estonian i-voting scheme vulnerable as in the original system because {\sf VerifApp} outputs $\mathcal{V}$'s intention explicitly. In order to overcome this vulnerability, each candidate will be displayed on the screen with a verification parameter and to let $\mathcal{V}$ checks these results with her own eyes using the parameters and decides whether the correct parameter and the candidate is matched. 

\subsection{Our Verification Mechanism}\label{lastproposedsystem}

We are now ready to describe our proposal. The actual flow of the protocol is similar to the Estonian i-voting system. In order to make the verification phase more resistant against the attacks described in Section \ref{potentialweakness}, we introduce a new parameter $q$ that will only be known to {\sf VoterApp} and {\sf VFS} where the size of $q$ is the same as the output size of the hash function $H$. More concretely, $\mathcal{V}$ chooses $q$ as a verification parameter and securely sends it to {\sf VFS} during the voting phase. Furthermore, during the verification phase, {\sf VFS} computes $\textsf{H}(E_{\textsf{asym}})$ and uses it as a symmetric key to encrypt $q$ for transmitting $\mathcal{E}_{\textsf{sym}}=\textsf{SymEnc}_{\textsf{H}(E_{\textsf{asym}})}(q)$ to {\sf VerifApp}. In order to let $\mathcal{V}$ check her vote, {\sf VerifApp} will perform a number of decryptions $q_i=\textsf{SymDec}_{\textsf{H}(E_{\textsf{asym}}^i)}(\mathcal{E}_{\textsf{sym}}), i=1,\cdots,m$ depending on the number of candidates. Hence, it outputs a list $Q=\{q_1,\cdots,q_m\}$ of possible verification parameters with the corresponding candidates. Finally, the voter $\mathcal{V}$ will manually check the output list on the {\sf VerifApp}'s screen  with her own eyes to learn the given $q$ during the voting phase on the position of the chosen candidate . The verification phase ends successfully if $q$ exists and its index is the same as the index of the candidate chosen by the voter $\mathcal{V}$ in the list $CL$. 


Here, the trick is to encrypt $q$ with the hash of the encrypted vote as the symmetric key (i.e., $\textsf{H}(E_{\textsf{asym}})$). In the verification phase, {\sf VerifApp} tries all candidates in order to generate the possible keys (hash of the possible encrypted votes). In order to complete the verification successfully, the index of the chosen candidate and the index of $q$ in $Q$ should match. Otherwise, either an alarm would be raised, or the voting procedure should be re-started, or the verification phase should be run in another device. It should be noted that manually finding the correct verification parameter with eyes is an important security measure so as not to reveal a voter's intention to {\sf VerifApp}. If the chosen candidate is displayed in plain form as in the Estonian i-voting scheme, adversaries may obtain a proof of their vote which may lead vote buying or coercion problems.  

The voting and the verification phases are explained in detail below.

\subsubsection{Voting Stage:}

\begin{enumerate}
	\item A voter $\mathcal{V}$: Authenticates to {\sf VFS} through {\sf VoterApp} using a national ID Card \footnote{Once the authentication (i.e, TLS) is complete, the TLS should encrypt all traffic between {\sf VFS} and {\sf VoterApp} to ensure that no data is leaked and to prevent man-in-the-middle attacks.}.
	\item {\sf VFS}: Sends $CL =\{c_1,\ldots,c_m\}$ to {\sf VoterApp} where $m$ is the number of candidates.
	\item $\mathcal{V}$: Chooses $c$ from $CL$ and 4 characters of random verification parameter (composed of 4 bytes per character) which is a random value $q_{\textsf{right}} \in_R \{0,1\}^{32}$.
	\item {\sf VoterApp}:
	\begin{enumerate}
		\item Generates a random number $r \in_R \{0,1\}^{\kappa}$ and $q_{\textsf{left}} \in_R \{0,1\}^{224}$. 
		\item Encrypts $c$ and $r$ by $pk_{S}$,  $E_{\textsf{asym}}=\textsf{AsymEnc}_{pk_{S}}(c,r)$.
		\item Signs $E_{\textsf{asym}}$ by $sk_{\mathcal{V}}$, $\textsf{SignEncVote}=\textsf{Sign}_{sk_{\mathcal{V}}}(E_{\textsf{asym}})$ .
		\item Sends $\textsf{SignEncVote}$ and the value $q=q_{\textsf{left}}||q_{\textsf{right}}$ to {\sf VFS}\footnote{We note that $q$ can also be encrypted (under the \textsf{VFS} public key) and signed by the voter to ensure its source and correctness.}.
	\end{enumerate}
	\item {\sf VFS}: 
	\begin{enumerate}
		\item Stores $\textsf{SignEncVote}$ and $q$. 
		\item Generates {\sf voteref}.
		\item Sends {\sf voteref} to $\mathcal{V}$ for verification phase.
	\end{enumerate}
\end{enumerate}

\begin{figure*}[htpb]
	\centering
	\caption{The voting phase of the proposed protocol}
	\includegraphics[width=1.0\linewidth]{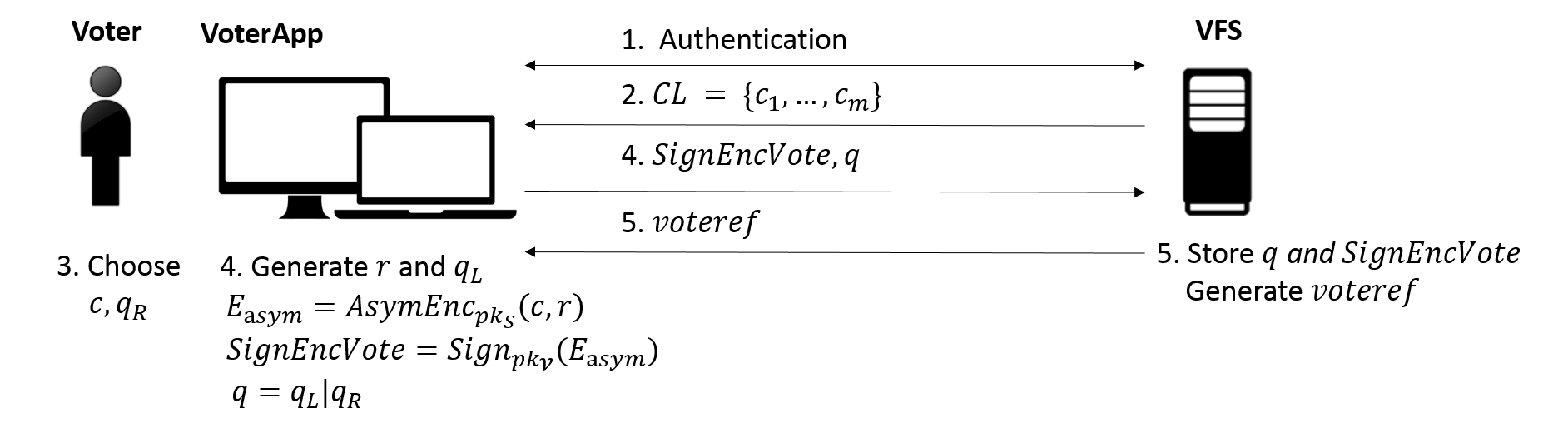}
	\label{fig:VotingPhase_new}
\end{figure*}

\subsubsection{Verification Stage:}

\begin{enumerate}
	\item \begin{enumerate}
		\item {\sf VoterApp}: Generates a {\sf QR code} including $r$ and {\sf voteref} and show on the screen. 
		\item {\sf VerifApp}: Scans the {\sf QR code} by the camera.
	\end{enumerate} 
	\item {\sf VerifApp}: Sends {\sf voteref} to {\sf VFS}.
	\item {\sf VFS}: 
	\begin{enumerate}
		\item Computes $\textsf{H}(E_{\textsf{asym}})$.
		\item Encrypts $\mathcal{E}_{\textsf{sym}}=\textsf{SymEnc}_{\textsf{H}(E_{\textsf{asym}})}(q)$. 
		\item Sends the ordered $m$-tuple $CL$ and $\mathcal{E}_{\textsf{sym}}$.
	\end{enumerate}
	\item {\sf VerifApp}: 
	\begin{enumerate}
		\item For each $c_j \in CL$, $j=1,\cdots,m$, computes; 
		\begin{enumerate}
			\item $E_{\textsf{asym}}^j = \textsf{AsymEnc}_{pk_{S}}(c_j, r)$. 
			\item the hash value $\textsf{H}(E_{\textsf{asym}}^j)$.
			\item $q_j = \textsf{SymDec}_{\textsf{H}(E_{\textsf{asym}}^j)}(\mathcal{E}_{\textsf{sym}})$.
		\end{enumerate}
		\item Shows the ordered $m$-tuple $Q=\{q_1, \cdots,q_m\}$ on the screen.
	\end{enumerate}
	
	\item $\mathcal{V}$: Finds the indices of $q \stackrel{?}{=}q_\alpha \in Q= \{q_1,\cdots,q_m\}$ and $c\stackrel{?}{=}c_\beta \in \{c_1,\cdots,c_m\}$, where $\alpha,\beta \in\{1,\cdots,m\}$ \footnote{As we described below in Section \ref{usability}, for usability concerns, only the last 32-bit integers $q_i$'s in $Q$ as 4 characters are viewed to the voter where $i = 1, \ldots, m$.}. 
	
	\begin{enumerate}
		\item Checks $\alpha \stackrel{?}{=} \beta$.
		\item If $\alpha = \beta$, the vote is received and stored correctly in {\sf VFS}.
		\item Else, $\mathcal{V}$ puts an alarm.
	\end{enumerate}
\end{enumerate}

\begin{figure*}[htpb]
	\begin{center}
	\caption{The Verification phase of the proposed protocol}
	\includegraphics[width=0.9\linewidth]{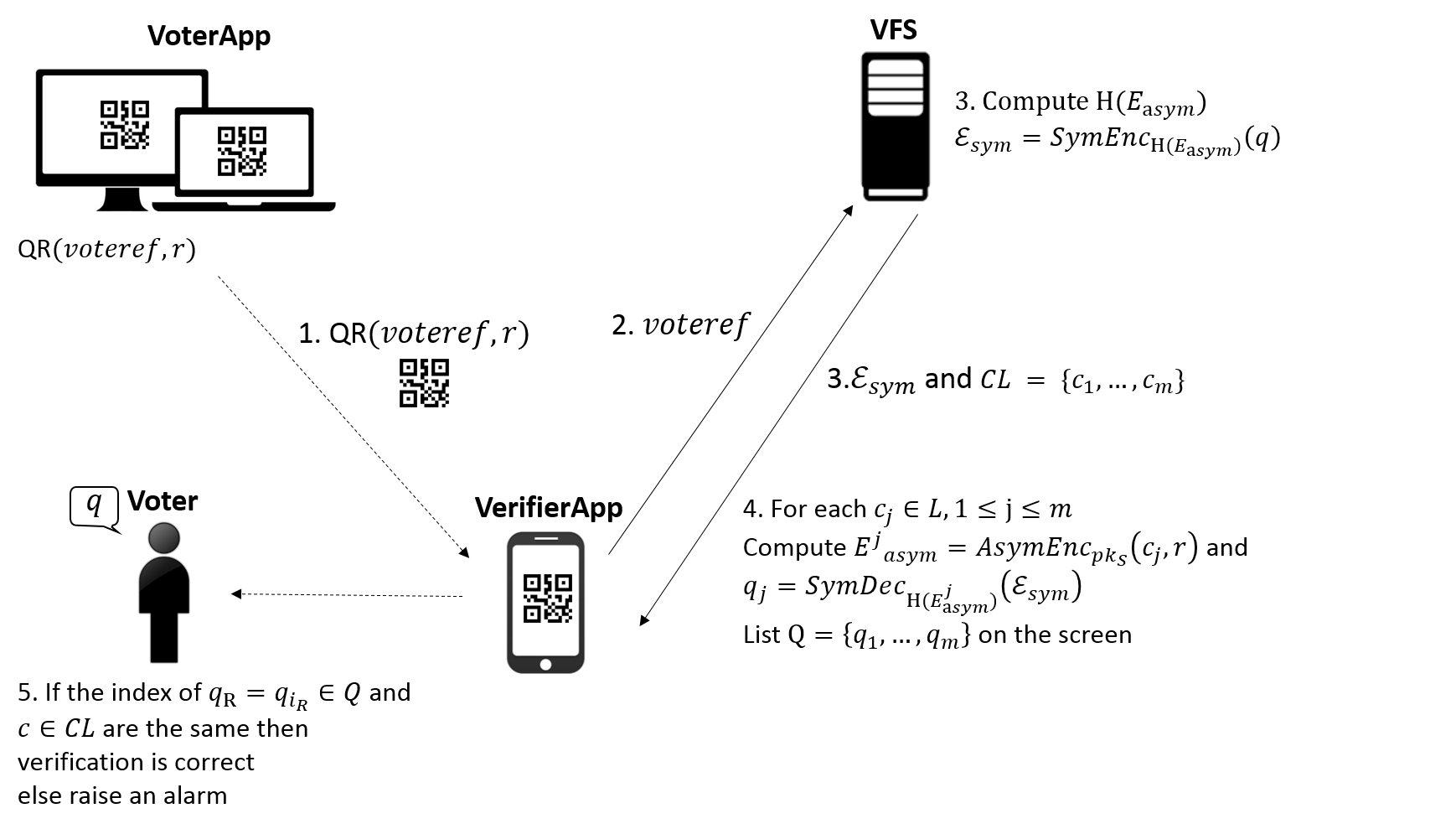}
	\label{fig:VerificationPhase_new}
	\end{center}
\end{figure*}

\section{Security Analysis of Our Verification Mechanism}\label{newsecurity}
We next show the correctness and prove the security of our mechanism. 

\begin{theorem} [Correctness]
	The verification phase of the proposed protocol ensures both the properties recorded-as-cast (the vote is stored in {\sf VFS}) and cast-as-intended (the vote reflects the intention of the voter).
\end{theorem}

\begin{proof}
During the verification phase, the {\sf VerifApp} first obtains $r$ and \textsf{voteref} from \textsf{VoterApp}. Next, \textsf{VerifApp} requests the relevant data according to the {\sf voteref} from {\sf VFS}. {\sf VFS} responds with the ordered $m$-tuple candidate list $CL=\{c_1,\cdots,c_m\}$ and $\mathcal{E}_{\textsf{sym}}=\textsf{SymEnc}_{\textsf{H}(E_{\textsf{asym}})}(q)$. Now, suppose that the voter $\mathcal{V}$ votes for the candidate $c_\beta \in CL$ where $\beta \in \{1, \cdots, m\}$. Then, {\sf VerifApp} completes the verification phase by computing $E_{\textsf{asym}}^i=\textsf{AsymEnc}_{pk_S}(c_i,r)$ for each $c_i \in CL$ and outputs the list $Q= \{q_1,\cdots,q_m\}$ where $ q_i = \textsf{SymDec}_{\textsf{H}(E_{\textsf{asym}}^i)}(\mathcal{E}_{\textsf{sym}}).$ Since the random value of the underlying asymmetric encryption $r$ is given to the {\sf VerifApp}, $E_{\textsf{asym}}^i$ becomes deterministic. Therefore, {\sf VerifApp} computes the correct value $q$ successfully. Finally, the voter $\mathcal{V}$ searches for $q$ (i.e., the last four characters of $q$ which were shown on the screen of the voter computer) on the verification device and finds the index $\beta \in \{1, \cdots, m\}$ where $q_{\beta}=q$. 

Therefore, this scheme guarantees that the ballot reflects the intention of the $\mathcal{V}$ (i.e., it has been cast-as-intended) and is correctly stored in {\sf VFS} (i.e., it has been recorded-as-cast). \qed
\end{proof}

\begin{theorem}[Privacy against malicious {\sf VoterApp} and its adversarial environment]
	Our verification mechanism described in Section \ref{lastproposedsystem} is secure against malicious {\sf VoterApp} or adversarial environment of {\sf VoterApp}.
\end{theorem}

\begin{proof}
Suppose {\sf VoterApp} (or its adversarial environment) cheats and tries to send the vote for another predefined candidate. Below, we show that the voter $\mathcal{V}$ can easily detect this adversarial action via \textsf{VerifApp}. 

Suppose a malicious {\sf VoterApp} tries to find the appropriate parameters to mock the {\sf VerifApp} by finding a value $r^*$ for the candidate $c^* \neq c \in CL$ satisfying $$q'=q'_{\textsf{left}}||q_{\textsf{right}} = \textsf{SymDec}_{\textsf{H}(E'_{\textsf{asym}})}$$ where $E'_{\textsf{asym}}=\textsf{AsymEnc}_{pk_S}(c^*, r^*)$ and $q=q_{\textsf{left}}||q_{\textsf{right}}$. Namely, a malicious {\sf VoterApp} can fool a voter by only showing the correct value $q_{\textsf{right}}$ which is in the position of the chosen candidate $c$.

In our system, for usability concerns, a voter chooses only four characters during the voting phase. However, a malicious \textsf{VerifApp} may guess $r^*$ in such a way that $q'=q'_{\textsf{left}}||q_{\textsf{right}} = \textsf{SymDec}_{\textsf{H}(E'_{\textsf{asym}})}(\mathcal{E}_{\textsf{sym}})$ where $E'_{\textsf{asym}}=\textsf{AsymEnc}_{pk_S}(c^*, r^*)$. However, this probability is $1/m \cdot \frac{1}{2^{32}}$ (choose the correct candidate of the voter and the corresponding value $r^*$). Because voting phase is independent for each voter, this probability goes to negligible for only 3 voters (i.e., $\frac{1}{2^{96}}$).

In the most general case, if a voter chooses the full length of $q$, then the computational cost of finding $r^*$ for certain $c, r$ and $c^*$ such that $\textsf{AsymEnc}_{pk_S}(c^*, r^*)=\textsf{AsymEnc}_{pk_S}(c, r)=E_{\textsf{asym}}$ is infeasible because of the underlying encryption scheme. An attacker may also simply guess $q^*$ which will be correct with probability at most $\frac{1}{2^k}$ where $k$ is the length of $q$. \qed
\end{proof}

\begin{theorem}[Privacy against malicious {\sf VerifApp} and its adversarial environment]
	Our verification mechanism described in Section \ref{lastproposedsystem} is secure against malicious {\sf VerifApp} or adversarial environment of the {\sf VerifApp} and does not leak any information about the $\mathcal{V}$ and her intention.
\end{theorem}

\begin{proof}
	
	For the verification phase, {\sf VerifApp} receives the parameters {\sf voteref}, $r$, $CL=\{c_1,\cdots,c_m\}$, and $\mathcal{E}_{\textsf{sym}}$. Additionally, it computes the lists $Q=\{q_1,\cdots,q_m\}$ and $\{E_{\textsf{asym}}^1,\cdots,E_{\textsf{asym}}^m\}$ where $m$ is the number of candidates. It can be easily checked if the received parameters or the computed values reveal any information about the $\mathcal{V}$'s intention. Therefore, an attacker sniffing {\sf VerifApp} can only learn whether a voter already has checked her vote.

	Note that {\sf VerifApp} never learns which $q_j$ is the correct $q$ since the $\mathcal{V}$ checks the index of verification parameter, which reveals the intention of the $\mathcal{V}$, by her own eyes from the ordered list $Q$. Therefore, an attacker sniffing {\sf VerifApp} should guess the correct $q$ with probability of $\frac{1}{m}$ which is no more than guessing the candidate randomly. Hence, {\sf VerifApp} leaks no information about the choice of the $\mathcal{V}$. \qed
\end{proof}

\begin{theorem}[Privacy against malicious {\sf VFS}]
	Our verification mechanism described in Section \ref{lastproposedsystem} is secure against a malicious {\sf VFS}.
\end{theorem}

\begin{proof}
	Our proposed verification mechanism uses the same voting phase with the Estonian i-voting system. The only difference is our mechanism is to the random verification parameter $q$ which gives no information about the vote. Therefore, the proof of the security of malicious {\sf VFS} is exactly the same with the Estonian scheme which is explained in \cite{EstSecAn-2010}.	\qed
\end{proof}

\section{Complexity and Usability Analysis} \label{complexity}

In Table \ref{comparison}, the complexity analysis of the proposed verification mechanism is tabulated. The cost of ridding the privacy leakage in the verification phase using symmetric key cryptography requires almost the same cost as the current Estonian i-voting scheme, specifically $1$ extra symmetric encryption for {\sf VFS} and $m$ extra symmetric encryptions for {\sf VerifApp}.

\begin{table}[!ht]
	\caption{Comparison of Vote Privacy and Computational Cost of the Verification Mechanisms where \textsf{SymEnc}, \textsf{AsymEnc}, \textsf{Sign} denote symmetric, asymmetric encryptions and digital signature.}
	\begin{center}
		\begin{tabular}{| p{2cm} | p{1.8cm} | p{1.5cm} || p{2.0cm} ||| p{3.2cm} |}
			\hline
			\centering
			&  &  &  & Vote Privacy \\
			& \hspace{0.2cm}{\sf \sf VoterApp} & \hspace{0.3cm} {\sf VFS} & \hspace{0.3cm} {\sf VerifApp} &  Against Corrupted \\
			& & & & Verification Device \\
			\hline  \hline
			\textbf{Estonian} & 1 \textsf{AsymEnc} & \hspace{0.5cm} $\emptyset$ & $m$ \textsf{AsymEnc} & \hspace{1.3cm}  \xmark\\ 
			\textbf{Verification} & + 1 \textsf{Sign} & & &  \\  \hline
			\textbf{This Paper} & 1 \textsf{AsymEnc} & 1 \textsf{SymEnc} & $m$ \textsf{AsymEnc} & \hspace{1.3cm} \checkmark \\ 
			& + 1 {\sf Sign} &  & + $m$ \textsf{SymEnc} &   \\ \hline	
		\end{tabular}
		\label{comparison}
	\end{center}
\end{table}

\subsection{Usability and Optimization Improvement for the Verification $q$}\label{usability}

The size of $q$ is important due to privacy and usability concerns. Firstly, because {\sf VerifApp} computes $q_i=\textsf{SymDec}_{\textsf{H}(E_{\textsf{asym}}^i)}(\mathcal{E}_{\textsf{sym}}), i=1,\cdots,m$ and outputs the list $Q=\{q_1,\cdots,q_m\}$, where each $q_i$ ($1\leq i \leq m$) must be random values of equal size and there exists a $\beta $ ($1 \leq \beta \leq m $) satisfying $q_\beta=q$. If the size does not match, the intended  can easily be detected. Furthermore, since $Q$ is disclosed to {\sf VerifApp}, each of the $q_i$'s must be indistinguishable in order to hide the intended vote and provide vote privacy.

\begin{table}[ht]
\centering
\caption{Shortened verification parameter on the screen on the voter computer.}
\label{fig:ParameterCheck}
	\begin{tabular}[h]{ | p{1.7cm} | p{0.2cm} |p{0.2cm} |p{0.2cm} |p{0.2cm} |p{0.4cm} |p{0.3cm} |p{0.3cm} |p{0.3cm} |p{0.3cm} |p{0.3cm} |p{0.3cm} |p{0.3cm} |p{0.3cm} |}	
		\hline 	
		Position $\#$ & 1 & 2 & 3 & 4 & $\cdots$ & 25 & 26 & 27 & 28 & \red{29} & \red{30} & \red{31} & \red{32} \\ \hline \hline
		$q$ & a & B & x & Q & $\cdots$ & q & K & j &  E &  \red{m} & \red{A} & \red{l} & \red{Q}\\ \hline 
		$q_{\textsf{right}}$ & &  & & & & & & & & \red{m} & \red{A} & \red{l} & \red{Q}\\ \hline			
	\end{tabular}
\end{table}

Note that, on the other side, it is not usable and impractical for a voter to generate a random $q$ of $t$-length bits. For example, if AES-256 and SHA3-256 are used as the encryption function then the voter must also generate a random $q$ of size 256-bits. Therefore, instead of generating a random value $q$, we proposed to generate $q=q_{\textsf{left}}||q_{\textsf{right}}$ where $q_{\textsf{left}} \in \{0,1\}^{224}$ is randomly generated by {\sf VoterApp} and $q_{\textsf{right}} \in_R \{0,1\}^{32}$ is provided by the $\mathcal{V}$. In this way, a voter can easily compare the value $q$ on \textsf{VoterApp} to the $q_i$'s on \textsf{VerifApp}. Hence, the protocol will run as defined in Section \ref{lastproposedsystem}, but both {\sf VoterApp} and {\sf VerifApp} will only display (for usability purposes) the last 32 bits of the verification parameters of $q$ and $q_1,\cdots,q_m$. A voter has to confirm with her own eyes that $q_{\textsf{right}}$ (i.e., the last 32 bits as 4 printable characters) is equal to the last 32 bits of $q_\beta$ exist in $Q$.

\begin{table}[ht]
	\centering
	\caption{The original $q$ and the respective $q_{\textsf{right}}$ verification values where $c_\beta$ is the intended vote.}
	\label{tab:findParameter}
	\begin{tabular}[h]{ | p{0.5cm} | p{0.5cm} || p{6.6cm} ||p{1.9cm} | }
		\hline
		\hspace{0.1cm}\textbf{$c_i$} & \hspace{0.1cm}\textbf{$q_i$} & \hspace{2.cm}\textbf{256-bits $q$} & \hspace{0.15cm}\textbf{$q_{\textsf{right}}$}\\
		\hline 	\hline
		\hspace{0.1cm}$c_1$ & \hspace{0.1cm}$q_1$ & qxvsEgaKMXpwApGDsNnPaNhjTJYt\red{qven} & \hspace{0.25cm} \red{qven} \\ \hline
		\hspace{0.2cm}\vdots & \hspace{0.2cm}\vdots & \hspace{2.9cm}\vdots & \hspace{0.7cm}\vdots \\ \hline 
		\hspace{0.1cm}$c_\beta$ & \hspace{0.1cm}$q_\beta$ & \hspace{0.15cm}aBxQwSOckfrzdYuaDNcvtTIDqKjE\red{mAlQ} & \hspace{0.25cm} \red{mAlQ}  \\ \hline
		\hspace{0.2cm}\vdots & \hspace{0.2cm}\vdots & \hspace{2.9cm}\vdots & \hspace{0.7cm}\vdots \\ \hline 
		\hspace{0.1cm}$c_m$ & \hspace{0.1cm}$q_m$ & \hspace{0.1cm} RatqPKvgtTAFectHpOeteDoPYKbT\red{kApp} & \hspace{0.25cm} \red{kApp} \\ \hline		
	\end{tabular}
\end{table}

As an illustration, we present a sample shortening operation in Table \ref{fig:ParameterCheck} for the parameter $q=$``aBxQwSOckfrzdYuaDNcvtTIDqKjEmAlQ''. Let $\ell=4$ and the position numbers are chosen as 8-bit words at the positions 29, 30, 31 and 32. Then, $q_{\textsf{right}}$ is shown as 'mAlQ' on the screen instead of whole $q$. In Table \ref{tab:findParameter}, Step 5 of the verification phase is presented as an example which shows the values on {\sf VerifApp} screen where $c_i$ is the $i$-th candidate and $q_i$ is the related verification value for $i=1,\cdots,m$.

\section{Conclusion}\label{conclusion}
Internet voting schemes are evolving and are already being practically used which aim to guarantee at least the same security level offered by the classical paper ballot voting systems. After the 2011 election in Estonia, a verification phase was integrated into the Estonian system to check whether the vote has been correctly recorded by avoiding possible attacks from malicious computers. However, this phase eventually displays the cast vote in plain form to the verifier application {\sf VerifApp} which may cause a serious privacy weakness and may compromise the fairness of the election. More concretely, during the verification phase of the Estonian voting system, displaying the vote in its plain form certainly may violate vote privacy. The problem comes from the fact that if the voter application {\sf VoterApp} running device is assumed to be compromised then the {\sf VerifApp} running device should also be assumed to be compromised. More concretely, there is no difference between trusting a device running the {\sf VoterApp} and trusting a device running the {\sf VerifApp}.

In this paper, we first point out this potential privacy issue of the Estonian verification mechanism in details. We further proposed a new and usable verification mechanism that does not disclose any information about the vote neither to the {\sf VerifApp} nor to the device where the {\sf VerifApp} runs. We show that our proposed verification system is strong against the aforementioned privacy weakness. As a future work, investigating a mechanism that resists the \textsf{Ghost Click Attack} and does not rely on additional post-channel communication would be interesting.

\end{document}